\documentclass[conference, 10pt, ﬁnal, letterpaper, twocolumn]{IEEEtran}
\usepackage[utf8]{inputenc}
\usepackage{stmaryrd}
\usepackage{amsfonts}
\usepackage{mathrsfs,amsmath}
\usepackage{amssymb}
\usepackage{makeidx}
\usepackage{mhchem}
\usepackage{multirow}
\usepackage{algorithmic} 
\usepackage{bm}
\usepackage{setspace}
\usepackage{amsthm}
\usepackage{empheq}
\usepackage{booktabs}
\usepackage{authblk}
\usepackage[ruled,linesnumbered,vlined]{algorithm2e}  
\usepackage{color}
\usepackage{enumerate}
\usepackage{makecell}
\usepackage[left=0.65in, right=0.65in, top=0.7in, bottom = 0.98in]{geometry}
\usepackage{tabularx,booktabs}
\usepackage{ifpdf}
 \ifpdf
 \else
 \fi
\usepackage{cite}
\usepackage{threeparttable}
\ifCLASSINFOpdf
   \usepackage[pdftex]{graphicx}
\else
\fi
\usepackage{amsmath}
\usepackage[percent]{overpic}

\usepackage{array} 
\usepackage{tikz,pgfplots,filecontents}
\usepackage{rotating}
\usetikzlibrary{spy} %
\pgfplotsset{width=9cm, height=6.8cm, compat=1.9}

\usepackage{physics}
\usepackage{amsmath}
\usepackage{tikz}
\usepackage{mathdots}
\usepackage{yhmath}
\usepackage{cancel}
\usepackage{color}
\usepackage{siunitx}
\usepackage{array}
\usepackage{multirow}
\usepackage{amssymb}
\usepackage{gensymb}
\usepackage{tabularx}
\usepackage{extarrows}
\usepackage{booktabs}
\usetikzlibrary{fadings}
\usetikzlibrary{patterns}
\usetikzlibrary{shadows.blur}
\usetikzlibrary{shapes}

\usetikzlibrary{calc}
\usepackage{pgfplotstable}
 
\ifCLASSOPTIONcompsoc 
  \usepackage[caption=false,font=normalsize,labelfont=sf,textfont=sf]{subfig}
\else
  \usepackage[caption=false,font=footnotesize]{subfig}
\fi

\usepackage{tikz-network}
\usetikzlibrary{shapes}
 
\usepackage{url}
\usepackage{hyperref}

\makeatletter
\let\NAT@parse\undefined
\makeatother

\newtheorem{theorem}{Theorem}

\IEEEoverridecommandlockouts   

\allowdisplaybreaks

\pgfplotsset{every axis legend/.style={%
cells={anchor=west},
inner xsep=3pt,inner ysep=2pt,nodes={inner sep=0.8pt,text depth=0.15em},
anchor=north east,%
shape=rectangle,%
fill=white,%
draw=black,
at={(0.98,0.98)},
font=\footnotesize,
}}

\pgfplotsset{every axis/.append style={line width=0.6pt,tick style={line width=0.8pt}}}

\begin{document}

\title{Less Carbon Footprint in Edge Computing by Joint Task Offloading and Energy Sharing}

\author{Zhanwei Yu\textsuperscript{1}, Yi Zhao\textsuperscript{1}, Tao Deng\textsuperscript{2}, Lei You\textsuperscript{3}, and Di Yuan\textsuperscript{1}}
\affil{\textsuperscript{1}Department of Information Technology, Uppsala University, Sweden
\authorcr \textsuperscript{2}School of Computer Science and Technology, Soochow University, China
\authorcr \textsuperscript{3}Department of Engineering Technology, Technical University of Denmark, Denmark
\authorcr {\em \textsuperscript{1}\{zhanwei.yu; yi.zhao; di.yuan\}@it.uu.se, \textsuperscript{2}dengtao@suda.edu.cn, \textsuperscript{3}lei.you@pm.me}}

\renewcommand*{\Affilfont}{\small}

\maketitle


\begin{abstract}
In sprite the state-of-the-art, significantly reducing carbon footprint (CF) in communications systems remains urgent. We address this challenge in the context of edge computing. The carbon intensity of electricity supply largely varies spatially as well as temporally. This, together with energy sharing via a battery management system (BMS), justifies the potential of CF-oriented task offloading, by redistributing the computational tasks in time and space. In this paper, we consider optimal task scheduling and offloading, as well as battery charging to minimize the total CF. We formulate this CF minimization problem as an integer linear programming model. However, we demonstrate that, via a graph-based reformulation, the problem can be cast as a minimum-cost flow problem. This finding reveals that global optimum can be admitted in polynomial time. Numerical results using real-world data show that optimization can reduce up to $83.3\%$ of the total CF.
\end{abstract}



\section{Introduction}
The current carbon footprint (CF) in information and communications technology (ICT) sector could be as high as $2.1\%$ to $3.9\%$ of the total figure \cite{Freitag2021Theclimate}. For Europe, for example, the baseline is $2\%$\cite{EuropeanCommission}. The authors of \cite{perrons2021digital} list four potential measures that can significantly help reduce the CF due to ICT, one of which is edge computing. 

The edge servers consume substantially less energy than conventional cloud data centers, though they need to be densely deployed. For green edge computing, research has been conducted for better design of networks \cite{mao2017survey}, including dynamic right-sizing, geographical load balancing, and the use of renewable energy. However, most works with respect to green edge computing consider energy consumption or efficiency as the objective function, rather than CF explicitly. 

Lower energy consumption contributes to reducing CF, however they are not equivalent. For example, in an edge computing network consisting of uniform servers and task distribution, considering only energy consumption would render task offloading useless (which itself would cost energy). However, as the CF intensity differs by time and space, offloading does help. We show a US domestic carbon intensity (CI) data set and another data set for three counties in Europe in Table \ref{tab:data} \cite{ElectricityMaps}. As we can see, there are large differences between spaces and time. For example, the CI in Oregon is five times higher than that in Washington. Also, The CI in Oregon at 24:00 is 31.7\% less than that at 08:00. The variation is due to the difference in the availability of various power sources over time and space. This motivates us to use the spatial and temporal information of CI to reduce CF by optimal task offloading and scheduling.  

\begin{table}[h]
    \caption{\label{tab:data}Two real-world data sets of carbon intensity.}
    \footnotesize
    \begin{center}
        \begin{threeparttable}[b]
            \begin{tabular}{*{5}{llrrr}}
                \toprule
                \midrule
                 \multicolumn{2}{c}{\multirow{2}{*}{\bf Region}} & \multicolumn{3}{c}{\bf Carbon intensity\tnote{1}}\\
                 \cmidrule(lr){3-5}
                \multicolumn{2}{c}{} & 08:00 & 16:00 & 24:00 \\
                \midrule
                 \multirow{3}{*}{\begin{turn}{90} USA \end{turn}} & Washington & 110 & 95  & 95\\
                 & Oregon & 605 & 579 & 413\\
                 & California & 325 & 292  & 238\\
                \midrule
                 \multirow{3}{*}{\begin{turn}{90} Europe \end{turn}} & Sweden & 24 & 26 & 25\\
                 & Germany & 375 & 285 & 381\\
                 & Poland & 593 & 573 & 547\\
                \midrule
                \bottomrule
            \end{tabular}
            
            \begin{tablenotes}
            	\footnotesize
            	\item $^1$Unit: g\ce{CO2}eq/kWh.
            \end{tablenotes}
        \end{threeparttable}
    \end{center}
\end{table}

Examining the current literature, we highlight the following contributions. The studies in \cite{van2012distributed, Rad2022Carbon, do2015proximal, aldossary2021towards, ahvar2021deca, yang2022carbon} have addressed CF in edge computing or fog computing. The authors of \cite{van2012distributed} provide a model that can estimate the CF in distributed data centers. Their results show that the total CF of a set of distributed small data centers is less than that from a big data center with equivalent compute capability. The authors of \cite{Rad2022Carbon} provide a Lyapunov-based algorithm for a distributed data center to minimize electricity cost subject to CF limit. In \cite{do2015proximal}, the authors consider minimizing the CF for video streaming in fog computing networks. In addition, the studies in \cite{aldossary2021towards} and \cite{ahvar2021deca} provide application placement methods to minimize the CF in fog computing networks. The authors of \cite{yang2022carbon} examine task scheduling policies to minimize the CF of edge computing networks via the drift-plus-penalty methodology in Lyapunov optimization. Table \ref{tab:difference} lists the features of our paper and the related works in \cite{van2012distributed, Rad2022Carbon, do2015proximal, aldossary2021towards, ahvar2021deca, yang2022carbon}. Note that the works in \cite{van2012distributed, Rad2022Carbon, do2015proximal, aldossary2021towards, ahvar2021deca} only consider optimizing either the spatial or temporal dimension instead of joint spatio-temporal optimization. 

\begin{table}[h] 
    \setlength\tabcolsep{3pt}
    \caption{\label{tab:difference}The difference among our paper and the related works.}
    \footnotesize
    \begin{center}
        \begin{threeparttable}[b]
            \begin{tabular}{*{7}{lccccc}}
                \toprule
                \midrule
                { \bf Work(s)} & \makecell[c]{  \bf Renewable\\   \bf energy} & \makecell[c]{  \bf Energy\\   \bf storage} & \makecell[c]{  \bf Spatio-temporal\\   \bf optimization} &\makecell[c]{  \bf Energy\\   \bf sharing} & \makecell[c]{  \bf Task\\   \bf offloading}\\
                \midrule
                  \cite{van2012distributed} & \checkmark & & & &\\
                  \cite{Rad2022Carbon, do2015proximal, aldossary2021towards, ahvar2021deca} & & & & & \checkmark \\ 
                  \cite{yang2022carbon} & & & \checkmark & & \checkmark \\
                   Our work & \checkmark & \checkmark & \checkmark & \checkmark& \checkmark\\
                \bottomrule
            \end{tabular}
        \end{threeparttable}
    \end{center}
\end{table}

To the best of our knowledge, the work in \cite{yang2022carbon} is the closest to ours, yet the differences are significant. The system model in our paper considers battery charging on top of task offloading and scheduling, making our scenario more comprehensive. To be more exact, our paper considers task scheduling and offloading, as well as energy sharing with battery charging to minimize the CF, utilizing temporal and spatial information of CI. The main contributions of this paper are as follows.
\begin{itemize}
    \item We consider an edge computing network with renewable energy sources and batteries, and model the CF for the resulting scenario. We consider task offloading and scheduling, as well as battery charging with energy sharing, with the objective of CF minimization. 
    \item Due to the discrete nature of task offloadling, the CF minimization problem leads to an integer linear programming (ILP) model. However, we reveal that the problem structure admits a reformulation using minimum-cost flow, implying that its global optimum can be computed in polynomial time. Thus our optimization approach is scalable. 
    \item We use real-world CI data to evaluate the performance. The numerical results show that, by optimal CF-aware task offloading and scheduling along with energy sharing, we can reduce up to $83.3\%$ of the total CF.
\end{itemize}


\section{System Model and Problem Formulation}

\input{fig-system_model}

\subsection{System Model}
Fig. \ref{fig:system_model} shows $S$ sites (illustrated by dashed rectangles) forming an edge computing network. Let $\mathcal{S} = \{1,2,...,S\}$. Each site consists of an edge computing server, a battery, a local renewable energy source, and the local power grid. The servers have the same specification and computing power. The batteries are connected to each other via a battery management system (BMS), hence the energy in a battery can be shared with the servers on the other sites\footnote{The BMS has been studies widely, see for example \cite{leithon2013online} and \cite{leithon2019task}.}. In this paper, we use the index of the site to refer to its components, e.g., server $s$, battery $s$, etc. More specifics of the system model are as follows;
\begin{itemize}
    \item A server can be powered by its local grid, its renewable energy source, and all the batteries. 
    \item A local battery can be charged by the local grid and the local renewable energy source. Additionally, the battery can provide energy to any other server. However, this comes with an energy transfer loss.
\end{itemize}

\subsection{Time Horizon}
We consider a scheduling horizon of $T$ time slots, denoted by $\mathcal{T} = \{1,2,...,T\}$. The set of tasks over the entire time horizon is represented by $\mathcal{N} = \{1, 2, ..., N\}$. In this paper, the tasks are assumed to be of the same type, i.e., the amount of energy to complete them is uniform, denoted by $E$. Without loss of generality, we use $E$ as our energy unit, and the relevant parameters (e.g., parameters $R$, $I$, $\alpha$, $\beta$, $L$, and $H$ to be introduced next) are all normalized by $E$, i.e., they are specified in multiples of $E$. 

Task $n$ is represented by a tuple $(o_n, d_n, s_n, \mathcal{S}_n)$, where $o_n, d_n \in \mathcal{T}, s_n \in \mathcal{S}, \mathcal{S}_n \subseteq \mathcal{S}$, such that
\begin{itemize}
    \item $o_n$ is the time slot when the task is generated;
    \item $d_n$ is the deadline of completing the task;
    \item $s_n$ is the server that the task is initially associated with;
    \item $\mathcal{S}_n$ represents the set of candidate servers\footnote{Some tasks might not be suitable for some servers because of, for example, latency due to the distance.} that can perform task $n$.
\end{itemize}
A task is to be completed in some time slot in the interval $[o_n, d_n]$.

Denote by $\pi_{nst}$ a binary variable that is one if and only if task $n$ is completed by server $s$ in time slot $t$. For any task, it needs to be completed once, thus we have
\begin{equation}\label{Constraint_1}
    \sum_{s \in \mathcal{S}} \sum^{d_n}_{t = o_n} \pi_{nst} = 1, \forall n \in \mathcal{N}.
\end{equation}

Due to normalization, a task requires one unit of energy. Thus in time slot $t$, the amount of energy consumed for task completion by server $s$ is $\sum_{n \in \mathcal{N}} \pi_{nst}$. Of the energy consumed, we use variables $x_{st}$, $y_{ss^\prime t}$, and $z_{st}$ to represent the amount of energy (again normalized by $E$) by server $s$ from the local grid, battery $s^\prime$, and the local renewable source, respectively, in time slot $t$. Then we have
\begin{equation}\label{Constraint_2}
     x_{st} + \sum_{s^\prime \in S}  y _{ss^\prime t} + z_{st} = \sum_{n \in \mathcal{N}} \pi_{nst}, \forall s \in \mathcal{S}, t \in \mathcal{T}.
\end{equation}
In addition, the maximum number of tasks that a server can perform in a time slot is constrained by the computing capacity, denoted by $H$. As a task amounts to one energy unit, we have the following constraint:
\begin{equation}\label{Constraint_3}
    x_{st}+\sum_{s^\prime \in\mathcal{S}} y_{ss^\prime t} +z_{st} \leq H, \forall s \in \mathcal{S}, t \in \mathcal{T}.
\end{equation}

We use variables $u_{st}$ and $v_{st}$ to represent the amount of energy from the local grid and the renewable source for charging, respectively. In addition, denote by variable $w_{st}$ the amount of remaining energy of battery $s$ at the end of time slot $t$. We use $L$ to denote the battery capacity. In time slot $t$, we have the following battery capacity constraint:
\begin{equation}\label{Constraint_4}
    u_{st} + v_{st} + w_{s(t-1)} \leq L, \forall s \in \mathcal{S}, t \in \mathcal{T} \setminus \{1\}.
\end{equation}
By the end of time slot $t$, the remaining energy $w_{st}$ is given by
\begin{equation}\label{Constraint_5}
    w_{st} = u_{st} + v_{st} + w_{s(t-1)} -\sum_{s^\prime \in \mathcal{S}} y_{ss^\prime t}, \forall s \in \mathcal{S}, t \in \mathcal{T} \setminus \{1\}.
\end{equation}

We use $R_{st}$ to denote the amount of renewable energy\footnote{Renewable energy can be predicted via for examples methods in \cite{li2019renewable}.} available from source $s$ in time slot $t$. The renewable energy is used for charging the local battery and supplying the local server, thus we have  
\begin{equation}\label{Constraint_6}
    z_{st} + v_{st} \leq R_{st}, \forall s \in \mathcal{S}, t \in \mathcal{T}.
\end{equation}

\subsection{Carbon Footprint and Mathematical Formulation} \label{subsec:CFA}

In the system model, CF occurs in three processes related to the grids. Denote by $I_{st}$ the CI of local grid $s$ in time slot $t$. The three types of CF can be calculated as follows.
\begin{enumerate}
    \item {\em Task completion}: CF will occur when a server consumes the energy from the local grid to complete some task. The total amount of CF of this can be expressed as
    \begin{equation}\label{CF_task_completion}
        \text{CF}^{\text{G}}(\boldsymbol{x}) = \sum_{t \in \mathcal{T}} \sum_{s\in \mathcal{S}} I_{st}x_{st}.
    \end{equation}
    \item {\em Battery charging}: The total amount of CF due to battery charging via the local grid can be calculated by
    \begin{equation}\label{CF_battery_charging}
        \text{CF}^{\text{B}}(\boldsymbol{u}) = \sum_{t \in \mathcal{T}} \sum_{s\in \mathcal{S}} I_{st}u_{st}.
    \end{equation}
    \item {\em Task offloading}: In this paper, we consider the worst-case CF in task offloading, i.e., we assume that the energy used for task offloading is entirely from the grid. We use $\alpha_{s_n s}$ to represent the required amount of the energy for transferring task $n$ from its initial server $s_n$ to $s$ ($\alpha_{s_n s} = 0$ if $s_n = s$). The total CF of task offloading in the network is given by
    \begin{align}\label{CF_task_offloading}
        \text{CF}^{\text{O}}(\boldsymbol{\pi}) =  \sum_{n \in \mathcal{N}} \sum_{s \in \mathcal{S}_n}  \sum^{d_n}_{t = o_n} \alpha_{s_n s} I_{st}  \pi_{nst}.
    \end{align}
\end{enumerate}
In addition, there will be some loss in energy transfer between a battery and the servers of other sites. We convert this loss into an equivalent amount of CF. Denote the loss per unit of energy from battery $s^\prime$ to server $s$ by $\beta_{s^\prime s}$ ($\beta_{s^\prime s} = 0$ if $s^\prime = s$). We assume that local grid $s$ provides energy to make up for the loss, and the total amount of the equivalent CF is thus 
\begin{equation}\label{CF_energy_loss}
    \text{CF}^{\text{L}}(\boldsymbol{y}) = \sum_{s \in \mathcal{S}}  \sum_{t \in \mathcal{T}}\sum_{s^\prime \in S_n}  \beta_{s^\prime s} I_{st} y _{ss^\prime t}.
\end{equation}

We consider minimizing the total CF in the edge computing network. The CF minimization problem can be formulated by integer linear programming (ILP) as follows,
\begin{subequations}\label{formulation}
    \begin{align}
         \mathop{\min}_{\boldsymbol{\pi} \in \{0,1\}, \atop \boldsymbol{x}, \boldsymbol{y}, \boldsymbol{z}, \boldsymbol{u}, \boldsymbol{v}, \boldsymbol{w} \geq 0} \ & \text{CF}^{\text{G}}(\boldsymbol{x}) +  \text{CF}^{\text{B}}(\boldsymbol{u}) + \text{CF}^{\text{O}}(\boldsymbol{\pi}) + \text{CF}^{\text{L}}(\boldsymbol{y}) \label{obj}\\
         \text{s.t.} & \quad \eqref{Constraint_1} \text{-} \eqref{Constraint_6},\notag
    \end{align}
\end{subequations}
where the objective function \eqref{obj} is the overall system CF.


\section{Problem Solving} 

\input{fig-networkflow}

Although the CF minimization problem is formulated as ILP in \eqref{formulation}, we will demonstrate that the global optimum can be computed in polynomial time. The idea is to construct a graph, along with entities for nodes and arcs, such that the problem maps to finding a minimum-cost flow in the graph. The graph construction, however, is non-trivial. We illustrate the concept in Fig. \ref{fig:networkflow} and detail the construction below.

\subsection{Overview}
A minimum-cost flow problem is an optimization problem of finding the cheapest possible way to route flow from supply node(s) to demand node(s) in a directed graph. In the graph, every arc has two attributes: 1) the per-unit flow cost, and 2) flow capacity. A feasible flow solution has to satisfy
\begin{enumerate}
    \item Flow balance constraint: The total incoming flow to a node plus the node's supply, if any, equals the total outgoing flow of the node plus the node's demand, if any. 
    \item Capacity constraint: The flow on every arc does not exceed the arc's capacity.
\end{enumerate}

In our graph shown in Fig \ref{fig:networkflow}, the green and yellow nodes are nodes with supply and demand, respectively; the other nodes are all transshipment nodes with no supply nor demand. In the graph, green $\theta$-nodes and node $\delta$, red $(\kappa, \rho)$ node pairs, blue $(\lambda, \varepsilon)$ node pairs, and yellow $\tau$-nodes represent the renewable sources, the merge of all the local power grids, batteries, servers, and tasks, respectively. The flows in the graph represent the amount of energy, and the per-unit cost of an arc is used for accounting for the CF.

Note that the per-unit costs and capacities of arcs default to zero and infinite, respectively, if we do not specify them. Also, we use $(a\rightarrow b)$ to represent the arc from node $a$ to node $b$. As can be seen, the graph consist in a number of sections, and in the following we detail each of them.

\subsection{The Entities in a Time Slot} 

There are $T$ sections indexed by time slot $t$ $(t\in \mathcal{T})$ in the graph. For each of them, there are three types of entities, namely the $\theta$-nodes, $(\kappa, \rho)$ node pairs, and $(\lambda, \varepsilon)$ node pairs.
\begin{enumerate}
    \item Green node $\theta^{t}_{s}$ represents renewable source $s$ in time slot $t$. Its supply in the graph is set to be $R_{st}$, i.e., the amount of energy available from source $s$ in time slot $t$. The meaning of the flows on arcs $(\theta^{t}_{s}\rightarrow \lambda^{t}_{s})$ and $(\theta^{t}_{s}\rightarrow \kappa^{t}_{s})$ are the same as those of the variables $z_{st}$ and $v_{st}$, respectively. By graph construction, the two flows adhere to constraint \eqref{Constraint_6}. Note that node $\theta^{t}_{s}$ also connects to the yellow demand node $\mu$ used for receiving the surplus flows (i.e., surplus energy) in the graph to guarantee the feasibility of the minimum-cost flow problem.
    \item A red node pair $(\kappa^{t}_{s}, \rho^{t}_{s})$ represents battery $s$ in time slot $t$. The capacity of arc $(\kappa^{t}_{s}\rightarrow \rho^{t}_{s})$ is set to be the capacity of battery $s$, i.e., $L$. The flows on arcs $(\delta \rightarrow \kappa^{t}_{s})$, $(\theta^{t}_{s} \rightarrow \kappa^{t}_{s})$, and $(\rho^{t-1}_{s}\rightarrow \kappa^{t}_{s})$ correspond to variables $u_{st}$, $v_{st}$, and $w_{s(t-1)}$, respectively, and they satisfy constraint \eqref{Constraint_4} due to flow balance.
    \item A blue node pair $(\lambda^{t}_{s}, \varepsilon^{t}_{s})$ represents server $s$ in time slot $t$. The flows on arcs $(\delta \rightarrow \lambda^{t}_{s})$, $(\rho^{t}_{s^\prime} \rightarrow \lambda^{t}_{s})$, and $(\theta^{t}_{s} \rightarrow \lambda^{t}_{s} )$ represent variables $x_{st}$, $y_{ss^\prime t}$, and $z_{st}$, respectively. The flows submit to constraint \eqref{Constraint_2} by flow balance. In addition, the capacity of arc $(\lambda^{t}_{s}\rightarrow\varepsilon^{t}_{s})$ is set to be the server capacity $H$, to model constraint \eqref{Constraint_3}. Note that the flow through arc $(\rho^{t}_{s^\prime}\rightarrow\lambda^{t}_{s})$ represents that battery $s^\prime$ provides energy to server $s$ in time slot $t$, therefore the per-unit cost of the arc is set to be $\beta_{s^\prime s} I_{st}$, and it accounts for the energy loss in task offloading. Clearly, the total flow cost of these arcs equals \eqref{CF_energy_loss}.
\end{enumerate}

\subsection{Battery Level Evolution}
Between any two neighboring time slots $t$ and $t+1$ in the graph, there are $S$ arcs, i.e., $(\rho^{t}_{s}\rightarrow\kappa^{t+1}_{s})$, $\forall s\in \mathcal{S}$. The arcs model the evolution of the battery energy between time slots $t$ and $t+1$. The flow on arc $(\rho^{t}_{s}\rightarrow \kappa^{t+1}_{s})$ is the remaining energy in battery $s$ at the end of $t$. By flow balance, the flow on arc $(\rho^{t}_{s}\rightarrow \kappa^{t+1}_{s})$ equals that the flow entering node $\rho^{t}_{s}$ minus the flows representing energy sharing on arcs $(\rho^{t}_{s}\rightarrow \lambda^{t}_{s^\prime})$, $\forall s^\prime \in \mathcal{S}$, and this flow balance implies constraint \eqref{Constraint_5}.

\subsection{The Grid}

The local power grids are merged into a green node $\delta$, and we identify the individual local grids via the arcs to the nodes corresponding to batteries and servers:
\begin{enumerate}
    \item Flow can be sent through arc $(\delta \rightarrow\kappa^{t}_{s})$, and this represents that local grid $s$ charges battery $s$ in time slot $t$. Clearly, its per-unit cost is $I_{st}$ to account for CF in battery charging, and the total flow cost of these arcs equals \eqref{CF_battery_charging}.
    \item We use the flow through arc $(\delta \rightarrow \lambda^{t}_{s})$ to represent that local grid $s$ provides energy to server $s$ in time slot $t$, and the unit flow cost of arc $(\delta \rightarrow \lambda^{t}_{s})$ is set to be $I_{st}$. The sum of flow costs of these arcs is equivalent to \eqref{CF_task_completion}.
\end{enumerate}
Arc $(\delta \rightarrow \mu)$ is used for routing the surplus flows from the grid. The flow on the arc represents the amount of grid energy not used by the edge computing network. The amount of supply of node $\delta$ is set to be a sufficiently large value, e.g., $N$. 

\subsection{The Tasks}

In the task section of the graph, a yellow demand node $\tau_n$ represents task $n$. The demand of node $\tau_n$ is one. The arcs between the $\varepsilon$-nodes and $\tau$-nodes are introduced based on the information of the tasks. Specifically, the presence of arc $(\varepsilon^{t}_{s}\rightarrow\tau_{n})$ represents that server $s$ can complete task $n$ in time slot $t$ ($o_n \leq t \leq d_n$). Thus, the flow on this arc represents variable $\pi_{nst}$, and the flow balance constraint for the task node is equivalent to constraint \eqref{Constraint_1}.
The per-unit cost of arc $(\varepsilon^{t}_{s}\rightarrow\tau_n)$ is set to be $\alpha_{s_{n}s}I_{s_{n}t}$, and the sum of flow costs on these arcs is equivalent to \eqref{CF_task_offloading}.

\subsection{Integrality and Complexity}

With the graph constructed, the minimum-cost flow problem can be
solved to the optimum by the network simplex
algorithm\cite{Ahuja1993NetworkFlows}, and we can obtain the optimal
solution to \eqref{formulation} by tracing the corresponding flows.
Note that the potential difficulty of \eqref{formulation} is that
$\boldsymbol{\pi}$ is binary. Although the flows are not required to
be integer, the integrality theorem of minimum-cost flow problems
guarantees that there always exists an integer optimal solution
\cite[Theorem 11.5]{Ahuja1993NetworkFlows}.  Given an integer optimal
solution, for each task node, only one unit of flow will be routed
from some $\varepsilon$-node as the demand of a task node is one,
i.e., a task will not be split between multiple servers at optimum.

In addition, we have the following theorem with respect to the
complexity of our optimization approach.

\begin{theorem}
CF minimization problem \eqref{formulation} can be solved in polynomial time.
\end{theorem}
\begin{proof}
The graph has in total $\mathcal{O}(ST + N)$ nodes and $\mathcal{O}(S^2 + NST)$ arcs, hence the graph constructing can be completed in polynomial time. The complexity of the network simplex algorithm using dynamic trees \cite{tarjan1997dynamic} for solving the problem of the graph is $\mathcal{O}((S^3T + NS^2 T^2+N^2ST)\log(ST + N)\log(STC + NC))$, where $C$ is the maximum cost of any arcs. Hence the conclusion.
\end{proof}


\section{Performance Evaluation}

In this section, we use a 24-hour CI data set of Sweden, Germany, and Poland from \cite{ElectricityMaps}; part of this data set is shown in Table \ref{tab:data}, for performance evaluation. In our simulation, the length of a time slot is an hour. The amount of renewable energy follows a binomial distribution $B(5, 0.5)$ \cite{coskun2011estimation} during daytime (7 am to 7 pm), and is zero otherwise. For any task $n$, $o_{n}$, $d_n$ ($o_n \leq d_n$), and $s_n$ all follow a discrete uniform distribution. In addition, we set $\mathcal{S}_n = \mathcal{S}$ $\forall n \in \mathcal{N}$, i.e., a task can be offloaded to any server. The other simulation parameters are specified in Table \ref{tab:parameters}. Note that the values of CF in our results are all normalized by $E$. 

\begin{table}[h]
    \caption{\label{tab:parameters}Simulation Parameters.}
    \footnotesize
    \begin{center}
        \begin{threeparttable}[b]
            \begin{tabular}{*{2}{lr}}
                \toprule
                \midrule
                {\bf Parameter} & {\bf Value}\\
                \midrule
                 The number of time slots ($T$) & $24$ \\
                 The number of servers ($S$) & $3$\\
                 The number of tasks ($N$) & $100$\\
                 The energy for transferring one task ($\alpha$) & $0.1$\\
                 Loss for transferring one energy unit ($\beta$) & $0.2$\\
                \bottomrule
            \end{tabular}
        \end{threeparttable}
    \end{center}
\end{table}

\pgfplotsset{compat=1.16}
\def\pgfplotsinvokeiflessthan#1#2#3#4{%
    \pgfkeysvalueof{/pgfplots/iflessthan/.@cmd}{#1}{#2}{#3}{#4}\pgfeov
}%
\def\pgfplotsmulticmpthree#1#2#3#4#5#6\do#7#8{%
    \pgfplotsset{float <}%
    \pgfplotsinvokeiflessthan{#1}{#4}{%
        #7%
    }{%
        \pgfplotsinvokeiflessthan{#4}{#1}{%
            #8%
        }{%
            \pgfplotsset{float <}%
            \pgfplotsinvokeiflessthan{#2}{#5}{%
                #7%
            }{%
                \pgfplotsinvokeiflessthan{#5}{#2}{%
                    #8%
                }{%
                    \pgfplotsset{float <}%
                    \pgfplotsinvokeiflessthan{#3}{#6}{%
                        #7%
                    }{%
                        #8%
                    }%
                }%
            }%
        }%
    }%
}%

\ifdefined\gconv
\else
\pgfmathsetmacro{\gconv}{0.1}
\fi

\pgfplotstableread[col sep=comma,header=true]{%
y,x,myvalue
0 ,  7 ,  27464.1
5 ,  7 ,  14408.1
10 ,  7 ,  9417.0
15 ,  7 ,  7024.5
20 ,  7 ,  6030.4
25 ,  7 ,  5610.0
30 ,  7 ,  5333.9
0 ,  8 ,  24431.0
5 ,  8 ,  13741.2
10 ,  8 ,  9063.1
15 ,  8 ,  6639.3
20 ,  8 ,  5808.4
25 ,  8 ,  5314.4
30 ,  8 ,  5094.4
0 ,  9 ,  22348.9
5 ,  9 ,  13279.9
10 ,  9 ,  8967.8
15 ,  9 ,  6499.8
20 ,  9 ,  5657.5
25 ,  9 ,  5182.6
30 ,  9 ,  4924.0
0 ,  10 ,  20266.8
5 ,  10 ,  12849.5
10 ,  10 ,  8885.2
15 ,  10 ,  6434.8
20 ,  10 ,  5536.5
25 ,  10 ,  5050.8
30 ,  10 ,  4829.8
0 ,  11 ,  19292.6
5 ,  11 ,  12518.0
10 ,  11 ,  8802.6
15 ,  11 ,  6402.3
20 ,  11 ,  5443.5
25 ,  11 ,  4975.0
30 ,  11 ,  4819.0
0 ,  12 ,  18799.8
5 ,  12 ,  12376.8
10 ,  12 ,  8720.0
15 ,  12 ,  6402.3
20 ,  12 ,  5378.5
25 ,  12 ,  4910.0
30 ,  12 ,  4808.2
0 ,  13 ,  18452.7
5 ,  13 ,  12268.8
10 ,  13 ,  8637.4
15 ,  13 ,  6402.3
20 ,  13 ,  5313.5
25 ,  13 ,  4877.5
30 ,  13 ,  4797.4
}{\datatable}
%
\pgfplotstablesort[create on use/sortkey/.style={
        create col/assign/.code={%
            \edef\entry{{\thisrow{x}}{\thisrow{y}}{\thisrow{myvalue}}}%
            \pgfkeyslet{/pgfplots/table/create col/next content}\entry
        }
    },
    sort key=sortkey,
    sort cmp={%
        iflessthan/.code args={#1#2#3#4}{%
            \edef\temp{#1#2}%
            \expandafter\pgfplotsmulticmpthree\temp\do{#3}{#4}%
        },
    },
    sort,
    columns/Mtx/.style={string type},
    columns/Kind/.style={string type},]\resulttable{\datatable}

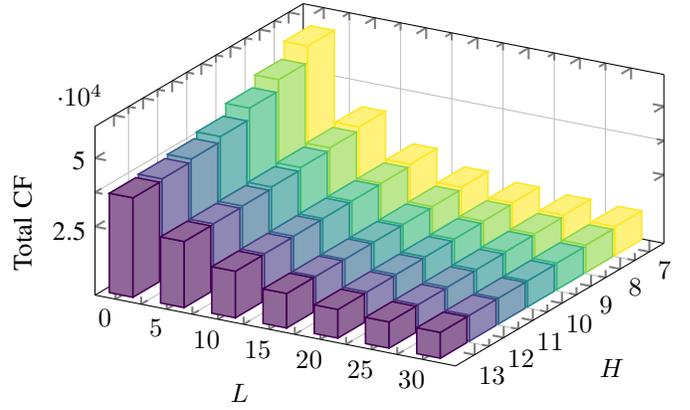
\begin{figure}[t]
	\begin{center}

\begin{tikzpicture}
\pgfplotsset{set layers}
\begin{axis}[
        view={120}{40},
        width=0.5 * \textwidth,
        height=0.35 * \textwidth,
        z buffer=none,
        xmin=6.5, xmax=13,
        ymin=-2, ymax=30,
        zmin=0, zmax=56000,
        enlargelimits=upper,
        ztick={25000, 50000},
        zticklabels={2.5, 5}, 
        xtick=data,
        extra tick style={grid=major},
        ytick=data,
        grid=minor,
        xlabel={$H$},
        ylabel={$L$},
        zlabel={Total CF},
        minor tick num=1,
        point meta=explicit,
        colormap name=viridis,
        scatter/use mapped color={
            draw=mapped color,fill=mapped color!60},
        execute at begin plot={}            
        ]
\path let \p1=($(axis cs:0,0,1)-(axis cs:0,0,0)$) in 
\pgfextra{\pgfmathsetmacro{\conv}{2*\y1}
\ifx\gconv\conv
\else
\xdef\gconv{\conv}
\typeout{Please\space recompile\space the\space file!}
\fi     
        };  
\path let \p1=($(axis cs:1,0,0)-(axis cs:0,0,0)$) in 
\pgfextra{\pgfmathsetmacro{\convx}{veclen(\x1,\y1)}
\typeout{One\space unit\space in\space x\space 
        direction\space is\space\convx pt}
        };                  
\path let \p1=($(axis cs:0,1,0)-(axis cs:0,0,0)$) in 
\pgfextra{\pgfmathsetmacro{\convy}{veclen(\x1,\y1)}
\typeout{One\space unit\space in\space y\space 
        direction\space is\space\convy pt}
        };                  
\addplot3 [visualization depends on={
\gconv*z \as \myz}, 
scatter/@pre marker code/.append style={/pgfplots/cube/size z=\myz},%
scatter/@pre marker code/.append style={/pgfplots/cube/size x=11.66135pt},%
scatter/@pre marker code/.append style={/pgfplots/cube/size y=9.10493pt},%
scatter,only marks,
mark=cube*,mark size=5,opacity=1]
 table[x expr={\thisrow{x}},y expr={\thisrow{y}},z
 expr={1*\thisrow{myvalue}},
 meta expr={-1*\thisrow{x}}
        ] \resulttable;
    \end{axis}
\makeatletter
\immediate\write\@mainaux{\xdef\string\gconv{\gconv}\relax}
\makeatother
\end{tikzpicture}
	\end{center}
	\caption{The 3D plot shows the total CF with respect to battery capacity $L$ and server capacity $H$.}\label{fig:result1}
\end{figure}

Fig. \ref{fig:result1} shows the performance results of our proposed scheme with respect to battery capacity $L$ and server capacity $H$. As we can see, there is a dramatic reduction of CF when the battery capacity goes from zero (i.e., no battery) to $5$ units. Thus the importance of battery (and BMS) to create an energy buffer for pursuing low CF is apparent. When $L > 5$, the total CF is hardly improved because such a capacity level can store all the energy required by the tasks. In addition, when the server capacity increases, the total CF decreases as more tasks can be offloaded to servers in low-CI regions.

We are curious about the impact of task offloading and energy sharing, respectively, hence we compare the performance of four schemes:
\begin{itemize}
    \item S1: This is the proposed scheme considering both task offloading and energy sharing via BMS in the network.
    \item S2: This scheme allows task offloading but the BMS is disabled, i.e., a battery can only provide energy to its local server.
    \item S3: This is the opposite to S2, namely energy sharing is enabled, but the tasks can not be offloaded among the servers.
    \item S4: In the last scheme, task offloading and energy sharing are both disabled; this corresponds to the most basic benchmark for comparison.
\end{itemize}
In the simulation, we obtain the performance of the four schemes by solving the corresponding minimum-cost flow problems. Figs. \ref{fig:result2} and \ref{fig:result3} show the performance results with respect to battery capacity $L$ and server capacity $H$, where we use w/ and w/o to represent with and without, respectively. Overall, compared with the conventional network (scheme S4), task offloading and energy sharing can help significantly reduce the total CF, up to $83.3\%$ by Fig.\ref{fig:result3} for $L = 5$ and $H = 16$.

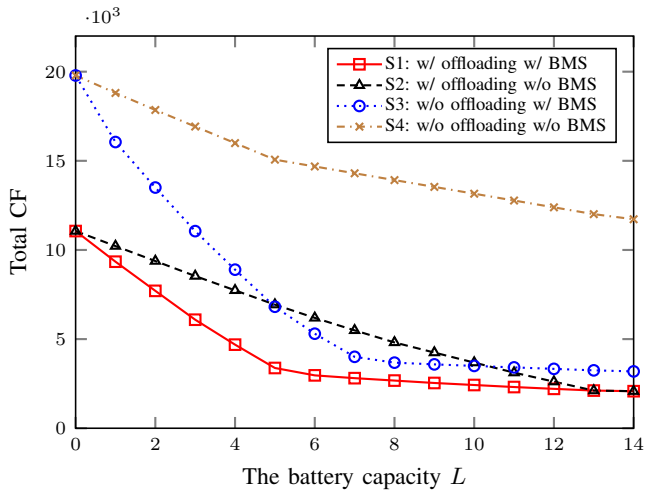
\begin{figure}
        \begin{center}
    		\begin{tikzpicture}
        		\begin{axis}[
        			scaled y ticks=base 10:-3,
        		    xlabel={The battery capacity $L$},
        		    ylabel={Total CF},
        		    xmin=0, xmax=14,
        		    ymin=0, ymax=22000,
        		    legend pos=north east,
        		    grid style=densely dashed,
        		    tick label style={font=\scriptsize},
        		    label style={font=\small},
        		    legend style={font=\scriptsize},
        		]
        		
            		\addplot[ color=red, mark=square, line width=0.8pt]     
            		coordinates { 
                    ( 0 , 11061.0 )
                    ( 1 , 9342.1 )
                    ( 2 , 7706.7 )
                    ( 3 , 6087.9 )
                    ( 4 , 4689.1 )
                    ( 5 , 3379.7 )
                    ( 6 , 2966.0 )
                    ( 7 , 2809.3 )
                    ( 8 , 2674.2 )
                    ( 9 , 2539.1 )
                    ( 10 , 2425.4 )
                    ( 11 , 2311.7 )
                    ( 12 , 2210.0 )
                    ( 13 , 2113.5 )
                    ( 14 , 2078.4 )
            		};

            		\addplot[ color=black, mark=triangle, densely dashed, mark options={solid}, line width=0.8pt]
            		coordinates {
                    ( 0 , 11061.0 )
                    ( 1 , 10218.3 )
                    ( 2 , 9378.0 )
                    ( 3 , 8537.7 )
                    ( 4 , 7735.5 )
                    ( 5 , 6933.3 )
                    ( 6 , 6190.4 )
                    ( 7 , 5490.7 )
                    ( 8 , 4812.6 )
                    ( 9 , 4245.1 )
                    ( 10 , 3682.4 )
                    ( 11 , 3119.7 )
                    ( 12 , 2614.0 )
                    ( 13 , 2113.5 )
                    ( 14 , 2078.4 )
            		};
            		
            		\addplot[ color=blue, mark=o, dotted, mark options={solid}, line width=0.8pt]
            		coordinates {
                    ( 0 , 19787.0 )
                    ( 1 , 16053.6 )
                    ( 2 , 13501.4 )
                    ( 3 , 11056.8 )
                    ( 4 , 8895.2 )
                    ( 5 , 6815.0 )
                    ( 6 , 5303.2 )
                    ( 7 , 4006.2 )
                    ( 8 , 3679.6 )
                    ( 9 , 3582.6 )
                    ( 10 , 3496.2 )
                    ( 11 , 3413.2 )
                    ( 12 , 3330.2 )
                    ( 13 , 3247.2 )
                    ( 14 , 3188.2 )
            		};

            		\addplot[ color=brown, mark=x, dash dot, mark options={solid}, line width=0.8pt]
            		coordinates {
                    ( 0 , 19787.0 )
                    ( 1 , 18807.0 )
                    ( 2 , 17852.0 )
                    ( 3 , 16923.0 )
                    ( 4 , 15994.0 )
                    ( 5 , 15065.0 )
                    ( 6 , 14683.0 )
                    ( 7 , 14301.0 )
                    ( 8 , 13919.0 )
                    ( 9 , 13537.0 )
                    ( 10 , 13155.0 )
                    ( 11 , 12773.0 )
                    ( 12 , 12391.0 )
                    ( 13 , 12009.0 )
                    ( 14 , 11723.0 )
            		};

            		\legend{S1: w/ offloading w/ BMS, S2: w/ offloading w/o BMS, S3: w/o offloading w/ BMS, S4: w/o offloading w/o BMS}
        		
        		\end{axis}
    		\end{tikzpicture}
    		
        \end{center}
    \caption{The total CF as function of the battery capacity $L$ with $H = 10$.}\label{fig:result2}
\end{figure}

\begin{figure}
        \begin{center}
    		\begin{tikzpicture}
        		\begin{axis}[
        			scaled y ticks=base 10:-3,
        		    xlabel={The server capacity $H$},
        		    ylabel={Total CF},
        		    xmin=8, xmax=17,
        		    legend style={at={(0.43, 0.75)},anchor=west},
        		    grid style=densely dashed,
        		    tick label style={font=\scriptsize},
        		    label style={font=\small},
        		    legend style={font=\scriptsize},
        		]
        		
            		\addplot[ color=red, mark=square, line width=0.8pt]     
            		coordinates { 
                    ( 8 , 4949.3 )
                    ( 9 , 4163.9 )
                    ( 10 , 3379.7 )
                    ( 11 , 3005.4 )
                    ( 12 , 2878.1 )
                    ( 13 , 2767.4 )
                    ( 14 , 2656.7 )
                    ( 15 , 2564.0 )
                    ( 16 , 2528.3 )
                    ( 17 , 2521.1 )
            		};
            		\addplot[ color=black, mark=triangle, densely dashed, mark options={solid}, line width=0.8pt]
            		coordinates {
                    ( 8 , 9506.6 )
                    ( 9 , 8014.9 )
                    ( 10 , 6933.3 )
                    ( 11 , 5927.9 )
                    ( 12 , 5033.1 )
                    ( 13 , 4248.9 )
                    ( 14 , 3464.7 )
                    ( 15 , 2766.0 )
                    ( 16 , 2528.3 )
                    ( 17 , 2521.1 )
            		};
            		
            		\addplot[ color=blue, mark=o, dotted, mark options={solid}, line width=0.8pt]
            		coordinates {
                    ( 8 , 6818.0 )
                    ( 9 , 6816.0 )
                    ( 10 , 6815.0 )
                    ( 11 , 6815.0 )
                    ( 12 , 6815.0 )
                    ( 13 , 6815.0 )
                    ( 14 , 6815.0 )
                    ( 15 , 6815.0 )
                    ( 16 , 6815.0 )
                    ( 17 , 6815.0 )
            		};
            		
            		\addplot[ color=brown, mark=x, dashed, mark options={solid}, line width=0.8pt]
            		coordinates {
                    ( 8 , 15091.0 )
                    ( 9 , 15065.0 )
                    ( 10 , 15065.0 )
                    ( 11 , 15065.0 )
                    ( 12 , 15065.0 )
                    ( 13 , 15065.0 )
                    ( 14 , 15065.0 )
                    ( 15 , 15065.0 )
                    ( 16 , 15065.0 )
                    ( 17 , 15065.0 )
            		};

            		\legend{S1: w/ offloading w/ BMS, S2: w/ offloading w/o BMS, S3: w/o offloading w/ BMS, S4: w/o offloading w/o BMS}
        		
        		\end{axis}
    		\end{tikzpicture}
    		
        \end{center}
    \caption{The total CF as function of the server capacity $H$ with $L = 10$.}\label{fig:result3}
\end{figure}
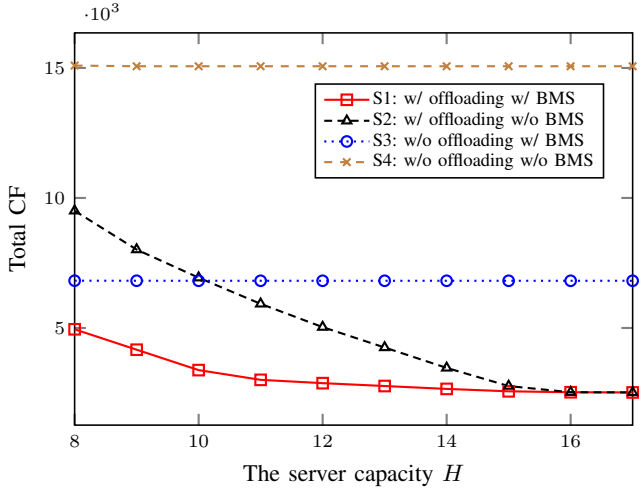

In Fig. \ref{fig:result2}, the comparison between S1 and S2 reveals that, when the battery capacity is either very low or very high, energy sharing is of little use. For the former case, it is because there is very little surplus energy available for sharing. For the latter case, the tasks can be offloaded to servers for which the local batteries are charged with low-CI energy. From the results of S1 and S2 in Fig. \ref{fig:result3}, energy sharing is not crucial when the server capacity is large. With high-capacitied servers, the tasks can be offloaded as much as possible to low-CI regions such that we do not need energy sharing for lower CF. 

In addition, the curves of S2 and S3 in Fig. \ref{fig:result2} show that task offloading and energy sharing almost give the same effect if the battery capacity is high. This is expected, as the operations of offloading a task to a low-CI region and obtaining some energy from that region achieve the same effect. Thus, joint optimization of offloading and energy sharing is of significance mainly if the server and battery capacities are moderate with respect to the number of tasks.

The performance of S3 and S4 in Fig. \ref{fig:result3} hardly change with increasing server capacity $H$, and this observation shows that optimizing only task scheduling in the time dimension without task offloading leads to a limited reduction of CF. The reason is that the difference of CI in the temporal dimension is typically not so large as that in the spatial dimension.


\section{Conclusion} \label{Sec:conclusion}
We have considered a CF minimization problem for edge computing by joint task offloading and energy sharing. We have formulated the problem as an integer linear programming model. Making the use of the structure of the problem, we have proposed a graph-based reformulation that enables to solve the problem to global optimum in polynomial time. The numerical results demonstrate the large potential of the use of spatial and temporal information of CI for reducing the CF, and joint task offloading and energy sharing does help significantly. Moreover, we can achieve the potential by either improving the energy storage capacity or computing capability. Future work includes, among other things, the extension of non-uniform energy consumption of tasks.

\bibliographystyle{IEEEtran}
\bibliography{mybibtex}

\begin{thebibliography}{10}
\providecommand{\url}[1]{#1}
\csname url@samestyle\endcsname
\providecommand{\newblock}{\relax}
\providecommand{\bibinfo}[2]{#2}
\providecommand{\BIBentrySTDinterwordspacing}{\spaceskip=0pt\relax}
\providecommand{\BIBentryALTinterwordstretchfactor}{4}
\providecommand{\BIBentryALTinterwordspacing}{\spaceskip=\fontdimen2\font plus
\BIBentryALTinterwordstretchfactor\fontdimen3\font minus
  \fontdimen4\font\relax}
\providecommand{\BIBforeignlanguage}[2]{{%
\expandafter\ifx\csname l@#1\endcsname\relax
\typeout{** WARNING: IEEEtran.bst: No hyphenation pattern has been}%
\typeout{** loaded for the language `#1'. Using the pattern for}%
\typeout{** the default language instead.}%
\else
\language=\csname l@#1\endcsname
\fi
#2}}
\providecommand{\BIBdecl}{\relax}
\BIBdecl

\bibitem{Freitag2021Theclimate}
\BIBentryALTinterwordspacing
C.~Freitag, M.~Berners-Lee, K.~Widdicks, B.~Knowles, G.~Blair, and A.~Friday,
  ``The climate impact of {ICT}: A review of estimates, trends and
  regulations,'' 2021. [Online]. Available:
  \url{https://arxiv.org/abs/2102.02622}
\BIBentrySTDinterwordspacing

\bibitem{EuropeanCommission}
E.~Commission and D.-G. for Communication, \emph{Supporting the green
  transition : shaping Europe’s digital future}.\hskip 1em plus 0.5em minus
  0.4em\relax Publications Office, 2020.

\bibitem{perrons2021digital}
R.~K. Perrons, ``How digital technologies can reduce greenhouse gas emissions
  in the energy sector's legacy assets,'' \emph{The Extractive Industries and
  Society}, vol.~8, no.~4, pp. 101\,010 (1--5), 2021.

\bibitem{mao2017survey}
Y.~Mao, C.~You, J.~Zhang, K.~Huang, and K.~B. Letaief, ``A survey on mobile
  edge computing: The communication perspective,'' \emph{IEEE communications
  surveys \& tutorials}, vol.~19, no.~4, pp. 2322--2358, 2017.

\bibitem{ElectricityMaps}
\BIBentryALTinterwordspacing
Electricity maps. [Online]. Available:
  \url{https://app.electricitymaps.com/map}
\BIBentrySTDinterwordspacing

\bibitem{van2012distributed}
W.~Van~Heddeghem, W.~Vereecken, D.~Colle, M.~Pickavet, and P.~Demeester,
  ``Distributed computing for carbon footprint reduction by exploiting
  low-footprint energy availability,'' \emph{Future Generation Computer
  Systems}, vol.~28, no.~2, pp. 405--414, 2012.

\bibitem{Rad2022Carbon}
A.~Radovanovic, R.~Koningstein, I.~Schneider, B.~Chen, A.~Duarte, B.~Roy,
  D.~Xiao, M.~Haridasan, P.~Hung, N.~Care, S.~Talukdar, E.~Mullen, K.~Smith,
  M.~Cottman, and W.~Cirne, ``Carbon-aware computing for datacenters,''
  \emph{IEEE Transactions on Power Systems}, pp. 1--1, 2022.

\bibitem{do2015proximal}
C.~T. Do, N.~H. Tran, C.~Pham, M.~G.~R. Alam, J.~H. Son, and C.~S. Hong, ``A
  proximal algorithm for joint resource allocation and minimizing carbon
  footprint in geo-distributed fog computing,'' in \emph{2015 International
  Conference on Information Networking (ICOIN)}.\hskip 1em plus 0.5em minus
  0.4em\relax IEEE, 2015, pp. 324--329.

\bibitem{aldossary2021towards}
M.~Aldossary and H.~A. Alharbi, ``Towards a green approach for minimizing
  carbon emissions in fog-cloud architecture,'' \emph{IEEE Access}, vol.~9, pp.
  131\,720--131\,732, 2021.

\bibitem{ahvar2021deca}
E.~Ahvar, S.~Ahvar, Z.~A. Mann, N.~Crespi, R.~Glitho, and J.~Garcia-Alfaro,
  ``{DECA}: A dynamic energy cost and carbon emission-efficient application
  placement method for edge clouds,'' \emph{IEEE Access}, vol.~9, pp.
  70\,192--70\,213, 2021.

\bibitem{yang2022carbon}
C.-S. Yang, C.-C. Huang-Fu, I.~Fu \emph{et~al.}, ``Carbon-neutralized task
  scheduling for green computing networks,'' \emph{arXiv preprint
  arXiv:2209.02198}, 2022.

\bibitem{leithon2013online}
J.~Leithon, T.~J. Lim, and S.~Sun, ``Online energy management strategies for
  base stations powered by the smart grid,'' in \emph{2013 IEEE International
  Conference on Smart Grid Communications (SmartGridComm)}.\hskip 1em plus
  0.5em minus 0.4em\relax IEEE, 2013, pp. 199--204.

\bibitem{leithon2019task}
J.~Leithon, L.~A. Su{\'a}rez, M.~M. Anis, and D.~N.~K. Jayakody, ``Task
  scheduling strategies for utility maximization in a renewable-powered iot
  node,'' \emph{IEEE Transactions on Green Communications and Networking},
  vol.~4, no.~2, pp. 542--555, 2019.

\bibitem{li2019renewable}
L.-L. Li, S.-Y. Wen, M.-L. Tseng, and C.-S. Wang, ``Renewable energy
  prediction: A novel short-term prediction model of photovoltaic output
  power,'' \emph{Journal of Cleaner Production}, vol. 228, pp. 359--375, 2019.

\bibitem{Ahuja1993NetworkFlows}
R.~K. Ahuja, T.~L. Magnanti, and J.~B. Orlin, \emph{Network Flows: Theory,
  Algorithms, and Applications}.\hskip 1em plus 0.5em minus 0.4em\relax
  Prentice-Hall, Inc., 1993.

\bibitem{tarjan1997dynamic}
R.~E. Tarjan, ``Dynamic trees as search trees via euler tours, applied to the
  network simplex algorithm,'' \emph{Mathematical Programming}, vol.~78, no.~2,
  pp. 169--177, 1997.

\bibitem{coskun2011estimation}
C.~Coskun, Z.~Oktay, and I.~Dincer, ``Estimation of monthly solar radiation
  distribution for solar energy system analysis,'' \emph{Energy}, vol.~36,
  no.~2, pp. 1319--1323, 2011.

\end{thebibliography}

\end{document}